\documentclass[aps,twocolumn,showkey,square,numbers,amssymb,amsmath,nobibnotes,nofootinbib]{revtex4-1}
\usepackage{amsmath,amsfonts,amssymb,amsthm,amscd,mathdots,bbm,graphicx,xcolor}
\usepackage{physics}

\usepackage{esint}
\usepackage{bm}
\usepackage{amsmath}
\usepackage{amsfonts}
\usepackage{amssymb}  
\usepackage{verbatim}
\usepackage{times}
\usepackage{graphicx}
\usepackage{mathrsfs}
\usepackage{dsfont}
\usepackage{titletoc} 
\usepackage{microtype}
\usepackage{hyperref}

\DeclareBoldMathCommand\boldlangle{\left\langle}
\DeclareBoldMathCommand\boldrangle{\right\rangle}

\newcommand{\barr}{\begin{eqnarray}}
\newcommand{\earr}{\end{eqnarray}}
\newcommand{\beq}{\begin{equation}}
\newcommand{\eeq}{\end{equation}}

\def\e{{\mathrm e}}

\newcommand{\be}{\begin{equation}}
\newcommand{\ee}{\end{equation}}

\newtheorem{thm}{Theorem}
\newtheorem{prop}{Proposition}

\newtheorem{lem}{Lemma}

\newtheorem{mydef}{Definition}
\theoremstyle{remark}
\newtheorem{rmk}{Remark}
\newtheorem{claim}{Claim}

\newenvironment{Proof}{\removelastskip\par\medskip
\noindent{\em Proof.}
\rm}{\penalty-20\null\par\medbreak}

\newcommand{\Hi}{\mathscr{H}}
\newcommand{\EM}{\mathcal{E}}
\newcommand{\IN}{\mathcal{I}}

\newcommand{\Dec}{\mathcal{D}}
\newcommand{\St}{\mathcal{S}}

\newcommand{\DIO}{\operatorname{DIO}}
\newcommand{\SIO}{\operatorname{SIO}}
\newcommand{\IO}{\operatorname{IO}}
\newcommand{\varphii}{\psi'}
\renewcommand{\leq}{\leqslant}
\renewcommand{\geq}{\geqslant}

\usepackage{graphicx,color}

\begin{document}

\title{Generic aspects of the resource theory of quantum coherence}

\author{Fabio Deelan Cunden$^{1,2}$, Paolo Facchi$^{3,4}$, Giuseppe Florio$^{4,5}$, Giovanni Gramegna$^{3,4}$}
\address{
$^{1}$ SISSA, Trieste 34136, Italy }
\address{
$^{2}$ Dipartimento di Matematica, Universit\`a di Bari, I-70125 Bari, Italy }
\address{
$^{3}$ Dipartimento di Fisica and MECENAS, Universit\`a di Bari, I-70126 Bari, Italy }
\address{
$^{4}$ INFN, Sezione di Bari, I-70126 Bari, Italy}
\address{
$^{5}$ Dipartimento di Meccanica, Matematica e Management, Politecnico di Bari, I-70125 Bari, Italy}
\date{\today}

\begin{abstract} 
	The class of incoherent operations induces a pre-order on the set of
quantum pure states, defined by the possibility of converting one state into the
other by transformations within the class.  We prove that if two $n$-dimensional
pure states are chosen independently according to the natural uniform
distribution, then the probability that they are comparable vanishes as $n\rightarrow\infty$. We also study the maximal success probability of incoherent conversions and find an explicit formula for its large-$n$ asymptotic distribution. Our analysis is based on the observation that the extreme values (largest and smallest components) of a random point uniformly sampled from the unit simplex are distributed asymptotically as certain explicit homogeneous Markov chains.
\end{abstract}

\maketitle

\section{Introduction}
The use of probability models and integral geometry to explain ``generic aspects
of quantum states'' is now a well-established point of view and there are multiple
lessons to learn from this approach~\cite{zyczkowskibeng,Hayden06,Popescu06,Bremner09,Facchi08Phase,DePasquale10,Cunden13,Gupta20,Gramegna20}. 
The next logical step is to use the same set of probabilistic ideas to
describe ``generic aspects of quantum resource theories''.
However, the
difficulties in describing the generic or typical aspects of resource theories
using probabilistic models remain considerable. 

Perhaps, the first question of
this flavour arose in the resource theory of entanglement.
After exposing a precise connection between the algebraic notion of
majorization \cite{Bhatia}  and convertibility among pure states by local
operations and classical  communications (LOCC), Nielsen~\cite{nielsen1999} made the remark that the set of pure states of a
bipartite system has a natural order relation, induced by the majorization
relation on their local spectra. This relation is \emph{not} total: not
all quantum states can be compared, i.e. are connected by a LOCC transformation.
Nielsen conjectured that for a bipartite system the set of pairs of pure states 
that are LOCC-convertible has relative volume asymptotically zero in the limit 
of large dimension: most quantum pure states are \emph{incomparable}!
He offered a probabilistic argument to justify the conjecture; shortly
after, another heuristic explanation based on integral geometry was put forward
in~\cite{Zyck02}. 

As far as we know, the only rigorous result around this
question is the proof that for an infinite-dimensional system, the set of pairs
of LOCC-convertible pure states is nowhere dense and so has measure
zero~\cite{Clifton02} (this statement though does not imply the conjecture).
Recently we have made some progress on Nielsen's conjecture by an extensive
numerical analysis of the volume of LOCC-convertible pairs of pure states~\cite{majorization}. 
The results support the conjecture, provide some nontrivial quantitative measure 
of the volume of LOCC-convertible states, and suggest new connections with random
matrix theory.
\par

It is natural to ask whether the property of `generic incomparability' is a
general feature  shared by other quantum resource theories.
Our attempt to answer this question starts from investigating this circle of
ideas for the class of \emph{quantum incoherent operations} \cite{Baumgratz14,
Biswas2017,Du15,Winter16,Chitambar16a,Chitambar16b,Zhu17,Fang18,Lami19,Regula20}. This choice is not arbitrary:
the resource theory of coherence is sufficiently simple to be tractable, and 
yet shares the connection with the algebraic notion of majorization that appears
in the  most interesting resource theories \cite{Horodecki13,Horodecki18}.
\par
  
It is the purpose of this article to present a complete analysis for the resource 
theory of coherence and indicate what might in the future be extended to other 
resource theories. 
\par
The structure of the paper is as follows. In Section~\ref{sec:resource} we
recall the definitions of incoherent and stricly incoherent quantum operations,
and the connection between incoherent convertibility and the majorization
relation.
In Section~\ref{sec:random} we
present the distributions of the smallest and largest `components' of random
pure quantum states; these are the main probabilistic properties relevant to our analysis.
Section~\ref{sec:results} contains the main results: the set of comparable
states in the resource theory of coherence has volume zero in the limit of large dimension
$n\to\infty$ 
(this is the analogue of Nielsen's conjecture in the theory of entanglement);
this problem is related to the persistence probability of a non-Markovian random
walk and we give numerical estimates on the rate of decay to zero; in the limit
$n\to\infty$ the
maximal success probability of incoherent conversion between two random
independent pure states has a nontrivial asymptotic distribution that we
characterise completely.
We conclude the paper with some final remarks in Section~\ref{sec:conclusion}.

\section{Resource theories of coherence}\label{sec:resource}
Recall that a \emph{resource theory} is defined by i) a set of  \emph{free states}, and ii) a class of \emph{free} or \emph{allowed operations}. In this work we consider the resource theories of coherence introduced and studied by \AA berg~\cite{Aberg06}, Baumgratz, Cramer, and Plenio~\cite{Baumgratz14}, Winter and Yang~\cite{Winter16}, and Chitambar and Gour~\cite{Chitambar16a,Chitambar16b}.

\subsection{Free states and free operations}
We denote by $\Hi_n$ a complex Hilbert space of dimension $n$, and by $\St_n$ the corresponding set of states $\rho$ (density matrices: $\rho\geq0$, $\operatorname{tr} \rho =1$).
Fix a basis $\{\ket{i}\}_{i=1}^n$ in $\Hi_n$, to be called the \emph{incoherent basis}. The choice may be dictated by physical considerations (for example, the eigenbasis of a particular observable). 

The set of free states in the resource theory of coherence is the set  of \emph{incoherent states} $\IN_n \subset \St_n$ defined as 
\begin{equation}
	\IN_n:=\biggl\{\rho\in\St_n\colon \rho=\sum_{i=1}^n p_i \ketbra{i}\biggr\},
\end{equation}
i.e.\ density matrices which are diagonal in the incoherent basis.
Notice that  $\IN_n$ is the image of $\St_n$ under the \emph{decohering map}, i.e. $\IN_n=\Dec(\St_n)$, where
\begin{equation}
	\Dec(\rho):=\sum_{i=1}^n \braket{i}{\rho|i} \ketbra{i}.
\end{equation}
The specification of the free states alone does not completely determine a resource theory. It is indeed necessary to specify a class of free operations. For the resource theory of coherence a number of different alternatives has been proposed, each yielding a different resource theory (see, e.g.,~\cite{Chitambar16a,Chitambar16b}). Here we focus on three possible choices of free operations, that allow for a criterion for convertibility between pure states in terms of the majorization relation. 

Recall that any quantum channel, that is a completely positive and trace preserving (CPTP) map $\EM:\St_n\rightarrow \St_n$ can be characterized in terms of a Kraus representation:
\begin{equation}
	\EM(\rho)= \sum_{\alpha} \mathcal{K}_\alpha (\rho) = \sum_{\alpha} K_\alpha \rho K_\alpha^\dagger,
\end{equation}
where   $\mathcal{K}_\alpha (\cdot) = K_\alpha (\cdot)  K_\alpha^\dagger$,
and $\{K_\alpha\}$ is a set of (non-uniquely determined) operators on $\Hi_n$, with $\sum_\alpha K_\alpha^\dagger K_\alpha = \mathbb{I}$.
 We can then define three classes of  CPTP maps on $\St_n$, representing three possible choices of free operations.
\begin{mydef}
	A
	quantum channel $\EM$ is said to be an \emph{incoherent operation} ($\IO$) if it can be represented by Kraus operators $\{K_\alpha\}$ such that
	\begin{equation}\label{eq:def1}
	\Dec(\mathcal{K}_\alpha(\ketbra{i}) )=\mathcal{K}_\alpha (\ketbra{i} ) 
	\end{equation}
	for all $\alpha$, and for all the elements $\ket{i}$ of the incoherent basis.
\end{mydef}
Note that if $\EM$ is an $\IO$, 
\begin{equation}
\rho\in\IN_n\Rightarrow	
\mathcal{K}_\alpha(\rho), 
\quad\text{for all  $\alpha$}.
\end{equation}
This restriction guarantees that, even if one has access to individual measurement outcomes $\alpha$ of the instrument $\{K_\alpha\}$, one cannot generate coherent states starting from an incoherent one. Notice that equation~\eqref{eq:def1} can be interpreted as a 
requirement of commutation between the decohering operation $\Dec$ and the operation $\mathcal{K}_\alpha$, 
when acting on the set of  incoherent states $\IN_n$. One can also restrict further the allowed operations by requiring the validity of such commutativity on  the whole set of states $\St_n$.
\begin{mydef}
	A 
	quantum channel $\EM$ is said to be a \emph{strictly incoherent operation} ($\SIO$) if it can be represented by Kraus operators $\{K_\alpha\}$ such that
	\begin{equation}
	\Dec(\mathcal{K}_\alpha(\rho)) = \mathcal{K}_\alpha(\Dec(\rho))
	\end{equation}
	for all $\alpha$, and for all  $\rho\in\St_n$.
\end{mydef}

%

One can also consider a third class of incoherent operations that satisfy the 
commutativity relation with $\Dec$ at `global'  level rather than at the level of Kraus operator representations. 

\begin{mydef}
	A quantum channel $\EM$ is said to be a \emph{dephasing covariant incoherent operation} ($\DIO$) if 
	\begin{equation}
		\Dec(\EM(\rho))=\EM(\Dec(\rho)) 
	\end{equation}
	for all $\rho\in\St_n$.
\end{mydef}

It is clear that  
$\SIO \subsetneq \IO$ and $\SIO\subsetneq \DIO$, while the classes $\IO$  and $\DIO$ are incomparable \cite{Chitambar16a,Chitambar16b}. 
It has been shown that transformations between \emph{pure states} (i.e.
rank-one projections $\psi = \ketbra{\psi}{\psi}$, with $\|\psi\|=1$) are fully
governed by the same majorization criteria \cite{Du15,Zhu17,Regula20}.
Therefore, although the three classes $\IO$, $\SIO$ and $\DIO$ are 
different from each other, they are operationally equivalent as far as
pure-to-pure state transformations are concerned. We also mention that all these classes are subclasses of the \emph{maximally incoherent operations} (MIO), which is the largest possible class of operations not generating coherent states starting from incoherent ones, and the very first to be studied \cite{Aberg06}. However, our results do not apply to this class, since (as far as we know) pure state conversions under MIO are not characterized by a majorization relation.
\subsection{Convertibility criterion and  majorization relation}
First we need to introduce some notation. In this paper $\Delta_{n-1}$ is the unit simplex, i.e. the set of $n$-dimensional probability vectors. For a vector $x$, we denote by $x^{\downarrow}$  the decreasing rearrangement of $x$, with $x^{\downarrow}_j \geq x^{\downarrow}_k$ for $j < k$. If $x,y$ are two vectors, we say that $x$ is \emph{majorized} by $y$ --- and write $x\prec y$ --- if \begin{equation}
\sum_{j=1}^k x^{\downarrow}_j\leq \sum_{j=1}^k y^{\downarrow}_j
\end{equation} 
for all $k=1,\dots,n$. For pure states $\psi =\ketbra{\psi}\in\St_n$, we write
\begin{equation}
	\delta(\psi):=(\abs{\psi_1}^2,\dots, \abs{\psi_n}^2)\in\Delta_{n-1},
\end{equation}
where $\psi_j = \braket{j}{\psi}$, i.e. the diagonal of the  density matrix $\psi$, in the (fixed) incoherent basis.

The following results expose the connection between the resource theories of coherence and the majorization relation.
\begin{thm}[\cite{Du15,Zhu17,Regula20}]\label{thm:conv} A pure state $\psi$ can be transformed into a pure state $\varphii$ under $\IO,\SIO$ or $\DIO$ if and only if $\delta(\psi)\prec \delta(\varphii)$.
\end{thm}
This theorem allows us to endow the set of pure states on $\Hi_n$ with a natural pre-order relation: we will write $\psi\prec \psi'$ whenever $\delta(\psi)\prec\delta(\psi')$. 


\begin{thm}[\cite{Zhu17}]\label{thm:probconv}
	For two pure states $\psi$ and $\varphii$, the maximal conversion probability 
	under $\IO$ is given by
	\be
	\Pi(\delta(\psi),\delta(\varphii)),
	\ee
	with $\Pi(\cdot,\cdot)$ being defined on pairs of probability vectors with $n$ nonzero components as
	\be
	\Pi(x,y)=\min_{1\leqslant k \leqslant n}\frac{\sum_{j=k}^n x^{\downarrow}_j}{\sum_{j=k}^n y^{\downarrow}_{j}}.
	\label{def:Pi}
	\ee
\end{thm}
Theorem \ref{thm:probconv} is the IO counterpart of an analogous result obtained by Vidal \cite{vidal1999} for LOCC conversions.
Note that $\Pi(\delta(\psi),\delta(\varphii))=1$ if and only if $\psi\prec\varphii$.
The theorem still holds if the class $\IO$ is
replaced by $\SIO$ or $\DIO$ as a consequence of Theorem \ref{thm:conv}.

Summing up, the three classes of incoherent operations considered are equivalent
for manipulation of pure states, and they are all governed by majorization
relations. In this paper we only consider pure state transformations. For
simplicity we will always refer to the class $\IO$, but all the results 
also hold for $\SIO$ and $\DIO$.

\section{Random pure states}\label{sec:random}
Let $\psi$ be a random pure state in $\St_n$ distributed according to the unitarily invariant measure. In the incoherent basis $\{\ket{i}\}$,
 $$
 \psi=\sum_{ij}\psi_{i}\overline{\psi_{j}}\ketbra{i}{j},
 $$
where $(\psi_1,\psi_2,\dots,\psi_n)$ is uniformly  distributed in the $n$-dimensional complex unit sphere, $\sum_j |\psi_j|^2 =1$.

Hence, the random vector $\delta(\psi)=(|\psi_1|^2,|\psi_2|^2,\dots,|\psi_n|^2)$ is \emph{uniformly distributed} in the simplex $\Delta_{n-1}$, 
\be
p_{\delta(\psi)}(x)=(n-1)! 1_{x\in\Delta_{n-1}}. 
\label{eq:pdelta}
\ee
 If $\mu$ is a uniform point in $\Delta_{n-1}$, i.e.\ distributed according to~\eqref{eq:pdelta},
  then the component $\mu_k$ is ``typically'' $O(1/n)$. The extreme components lie 
  instead on very different scales. The largest components $\mu^{\downarrow}_1, \mu^{\downarrow}_2,\ldots$ 
  are of size $\log n/n$ with fluctuations of  $O(1/n)$; the smallest components $\mu^{\downarrow}_n, \mu^{\downarrow}_{n-1},\ldots$ are on the much smaller scale $1/n^2$, with fluctuations of  $O(1/n^2)$.
We now give a precise asymptotic descriptions of the extreme statistics of the
uniform distribution on $\Delta_{n-1}$: they are distributed as time-homogenous
Markov chains with explicit (and remarkably simple) transition densities. We
must say that the uniform distribution on the simplex is one of the favourite
topics in geometric/integral probability~\cite{Baci20}, and its relevance in
quantum applications has been already highlighted in the past~\cite{zyczkowskibeng,Majumdar08}.  The following result is probably folklore but we could not trace it in the literature. We report it here since it is crucial for the next analysis. The proof is given in Appendix~\ref{app:proof_propMarkov}.
    \begin{prop}
  \label{prop:Markov}
   Let $\mu=(\mu_1,\mu_2,\ldots,\mu_n)$ be a uniform point in $\Delta_{n-1}$. Denote by $\mu^{\downarrow}$ the decreasing rearrangement of $\mu$. Then, for any fixed integer $k\geq1$, the following hold as $n\to\infty$:
   \begin{itemize}
   \item[(i)] the rescaled vector  of the smallest components $
  \bigl(n^2\mu^{\downarrow}_{n-j+1}\bigr)_{1\leq j\leq k}
  $
  converges in distribution to 
  $
  (V_1,V_2,\dots,V_k)
  $,
  where $(V_j)_{j\geq1}$ is a  Markov chain with initial  and transition densities given by
\be
  	\begin{gathered}\label{eq:V_transit}
  	f_{V_1}(v)=\exp(-v)1_{v\geq0},\\
	f_{V_{j+1}|V_j}(u|v)=\exp(v-u)1_{u\geq v};
	\end{gathered}
\ee

 \item[(ii)]   the rescaled vector   of the largest components $
  \bigl(n\mu^{\downarrow}_j-\log n\bigr)_{1\leq j\leq k}
  $
  converges in distribution to 
  $
  (W_1,W_2,\dots,W_k)
  $,
where $(W_j)_{j\geq1}$ is a Markov chain with initial  and transition densities given by
\be
  \begin{gathered}\label{eq:W_transit}
f_{W_1}(w)=\exp(-\e^{-w}-w),\\
f_{W_{j+1}|W_j}(u|w)=\exp(\e^{-w}-\e^{-u}-u)1_{u\leq w}.
\end{gathered}
\ee
\end{itemize}
    \end{prop}
 \par
Note that, by the Markov property, we can write the joint density of  $(V_1,\dots,V_k)$,
\be
f_{V_1,\dots,V_k}(v_1,\dots,v_k)=\exp\left(-v_k\right)1_{0\leq v_1\leq v_2\leq\cdots\leq v_k},
\label{eq:joint_V}
\ee
and the joint density of $(W_1,\dots,W_k)$,
\begin{align}\notag
	&f_{W_1,\dots,W_k}(w_1,\dots,w_k)=\\
	&\quad=\exp\left(-w_1-w_2-\cdots-w_k-\e^{-w_k}\right)1_{w_1\geq w_2\geq\cdots\geq w_k}.
	\label{eq:joint_W}
\end{align}

The next Lemma gives a concrete realization of the Markov chains $(V_j)_{j\geq1}$ and $(W_j)_{j\geq1}$ in terms of discrete-time continuous random walks. In the following, $(X_j)_{j\geq1}$ is a sequence of independent exponential random variables with rate $1$, i.e. $P(X_j \leq x)=1-\exp(-x)$.
 \begin{lem}
 \label{lem:Markov} Let  $(V_j)_{j\geq1}$ and  $(W_j)_{j\geq1}$ be the Markov chains defined in~\eqref{eq:V_transit} and~\eqref{eq:W_transit}, respectively. Then,
 \label{lem:repr}
\begin{align}
  (V_j)_{j\geq1}&\stackrel{D}{=} \left(X_1+\dots+X_j\right)_{j\geq1},
   \label{eq:repr_V}\\
 (W_j)_{j\geq1}&\stackrel{D}{=} \left(-\log\left(X_1+\dots+X_j\right)\right)_{j\geq1},
  \label{eq:repr_W}
\end{align}
where $\stackrel{D}{=}$ denotes equality in distribution.
 \end{lem}
 See Fig.~\ref{fig:cartoon} for a pictorial illustration of Proposition~\ref{prop:Markov} and Lemma~\ref{lem:Markov}.
 
 \begin{figure}[t]
	\includegraphics[width=1\columnwidth]{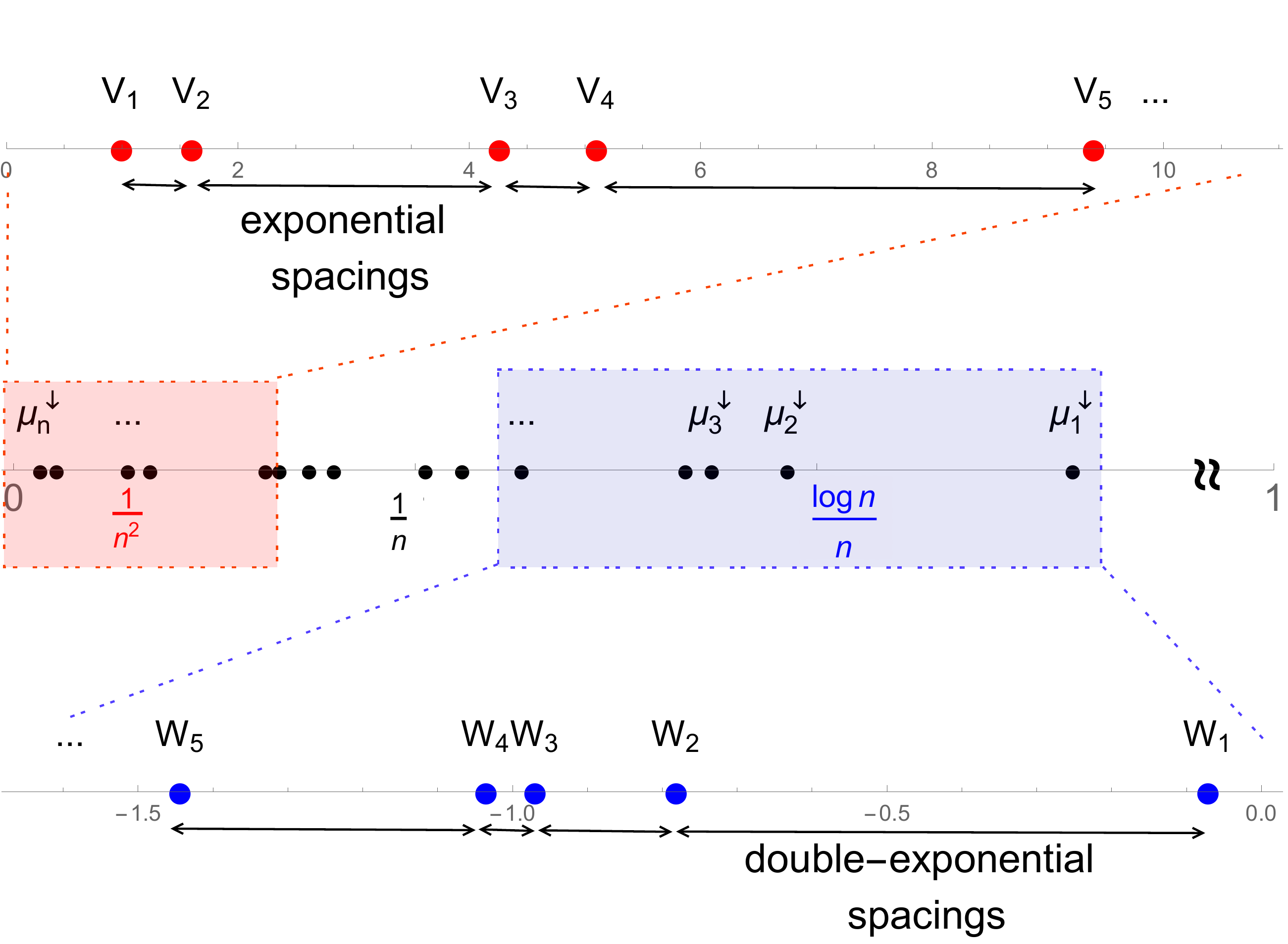}
	\caption{MIDDLE: The $n$ components of a uniform point $\mu$ in $\Delta_{n-1}$ lie in the unit interval $[0,1]$. TOP: For large $n$, the smallest components $\mu^{\downarrow}_n,\mu^{\downarrow}_{n-1},\dots,\mu^{\downarrow}_{n-k+1}$ after a proper rescaling behave statistically as the first $k$ points of a Poisson process $(V_{j})_{j\geq1}$ with exponential spacings. BOTTOM: The largest components $\mu^{\downarrow}_1,\mu^{\downarrow}_{2},\dots,\mu^{\downarrow}_{k}$ after a proper rescaling behave statistically as the first $k$ points of a Poisson process $(W_{j})_{j\geq1}$ with double-exponential (or Gumbel) spacings.} 
	\label{fig:cartoon}
\end{figure}
    
 \section{Volume of the set of IO-convertible states}
 \label{sec:results}

  In 1999, Nielsen~\cite{nielsen1999} conjectured that the relative volume of pairs
of LOCC-convertible bipartite pure states vanishes in the limit of large
dimensions. The precise statement of the conjecture is that for two independent
random points in the simplex with a distribution of random matrix type (see
Sec.~\ref{sub:Entanglement} below), the probability that they are in majorization relation is asymptotically zero. Here we pose a similar question in the theory of coherence: is it true that `most' pairs of pure $n$-dimensional quantum states are \emph{not} IO-convertible if $n$ is large? The answer is `Yes'.
  \subsection{Asymptotics $n\to\infty$}
\begin{thm}\label{thm:Pvanishing}
Let $\psi$ and $\varphii$ be independent random pure states in $\St_n$. Then,
$$
\lim_{n\to\infty} P\left(\psi\prec\varphii\right)=0.
$$
 \end{thm}
 \begin{Proof}  We use the shorter notation $\mu:=\delta(\psi)$ and $\mu':=\delta(\varphii)$. 
It turns out to be convenient to write the majorization relation $\mu\prec\mu'$ as
 \be
 \sum_{i=n-j+1}^{n}\mu^{\downarrow}_{i}\geq  \sum_{i=n-j+1}^{n}\mu'^{\downarrow}_{i}, \,\text{for all $ j= 1, \dots, n$},
 \label{eq:maj}
 \ee
 by using  the normalization condition $\sum_{i=1}^n \mu^{\downarrow}_{i} =\sum_{i=1}^n \mu'^{\downarrow}_{i}$.
 The idea of the proof,  inspired by~\cite{Pittel99}, is to show that 
 $$
 \lim_{k\to\infty}\lim_{n\to\infty}P\left(\text{$\mu,\mu'$ meet the first $k$ conditions in~\eqref{eq:maj}}\right)=0.
 $$
From Proposition~\ref{prop:Markov},  
$$
\text{$(\mu^{\downarrow}_{n},\mu^{\downarrow}_{n-1},\ldots,\mu^{\downarrow}_{n-k+1})$ is asymptotic to $\left(n^{-2}V_j\right)_{1\leq j\leq k}$},
$$
as $n\to\infty$.
By  Lemma~\ref{lem:Markov}, we have the representation  $(V_j)_{j\geq1}\stackrel{D}{=}\left(X_1+\dots+X_j\right)_{j\geq1}$. 
 Analogous representation holds for the $k$ smallest components of $\mu'$ with their own sequence  $(V'_j)_{j\geq1}\stackrel{D}{=}\left(X'_1+\dots+X'_j\right)_{j\geq1}$. 
 Consider the probabilities ($1\leq k\leq n$)
$$
\pi_{n,k}:=  P\Biggl(
\sum_{i=n-j+1}^n\mu^{\downarrow}_{i}\geq \sum_{i=n-j+1}^n\mu'^{\downarrow}_{i},\,\text{for all $1\leq  j\leq k$}\Biggr).
$$
 Of course,
$$
 \pi_{n,n}=P\left(\mu\prec\mu'\right),\quad\text{and}\quad  
 \pi_{n,n}\leq \pi_{n,k}.
$$ 
By taking the limit $n\to\infty$  we get

\begin{align*}
p_k:=\lim_{n\to\infty}\pi_{n,k}
&=P\left(\sum_{i=1}^j V_i\geq \sum_{i=1}^j V_i',
\text{for all $1\leq  j\leq k$}\right).
\end{align*}
It is clear that 
$$
0\leq\limsup_n\pi_{n,n}\leq\lim_{k\to\infty}p_k.
$$
Hence, to prove that $\pi_{n,n}\to0$ as $n\to\infty$, it is enough to show that $p_k\to0$ as $k\to\infty$.  
The sequence $\widetilde{V}_k:=(V_k-V'_k)=\sum_{j=1}^k\widetilde{X}_j$, $k\geq1$ is a time-discrete continuous random walk with independent steps $\widetilde{X}_j:=(X_j-X'_j)$ distributed according to the two-side exponential density $(1/2)\exp(-|x|)$. The process $I_k:=\sum_{j=1}^k\widetilde{V}_j$, $k\geq1$ is the corresponding \emph{integrated random walk} (IRW). Hence, $p_k$ is the so-called \emph{persistence probability} above the origin of the IRW,
\be
p_k=P\left(\min_{1\leq j\leq k}I_j\geq0\right).
\label{eq:pers_p_k}
\ee
The proof that the persistence probability of the IRW asymptotically vanishes, 
\be
p_k\to0, \quad \text{as } k\to\infty,
\label{eq:ppvanish}
\ee
follows from the Lindeberg-Feller central limit theorem and the Kolmogorov 0-1
law, and is given in Appendix~\ref{app:proof_claim}.
 \hfill\qed
 \end{Proof}
It might seem that, having to deal with i.i.d. variables $\widetilde{X}_j$,
the proof that $p_k\to0$ is straightforward. Note however, that the integrated
random walk $(I_k)_{k\geq1}$ is not Markov ($I_k$ depends on all variables
$\widetilde{X}_j$, $j\leq k$) and this explains why some analysis is required.

We should also mention that \eqref{eq:ppvanish} is a direct consequence of several
persistence results for IRW~\cite{Dembo13,Vysotsky07,Vysotsky14}, i.e.
asymptotic estimates for the sequence $p_k$ in
\eqref{eq:pers_p_k}. For our purposes however, we do not
really need the precision of those asymptotic statements and this is the reason
for including in Appendix~\ref{app:proof_claim}  a proof of \eqref{eq:ppvanish} based on elementary probability.

The proof strategy in Theorem~\ref{thm:Pvanishing} is based on bounding $\pi_{n,n}$ by a sequence $p_k$ independent on $n$, and therefore gives no information on the rate  of decay of $\pi_{n,n}$ to zero. Some  insights can be obtained from the perspective of persistence probabilities as discussed in the next section.
\begin{figure}[t]
	\includegraphics[width=1\columnwidth]{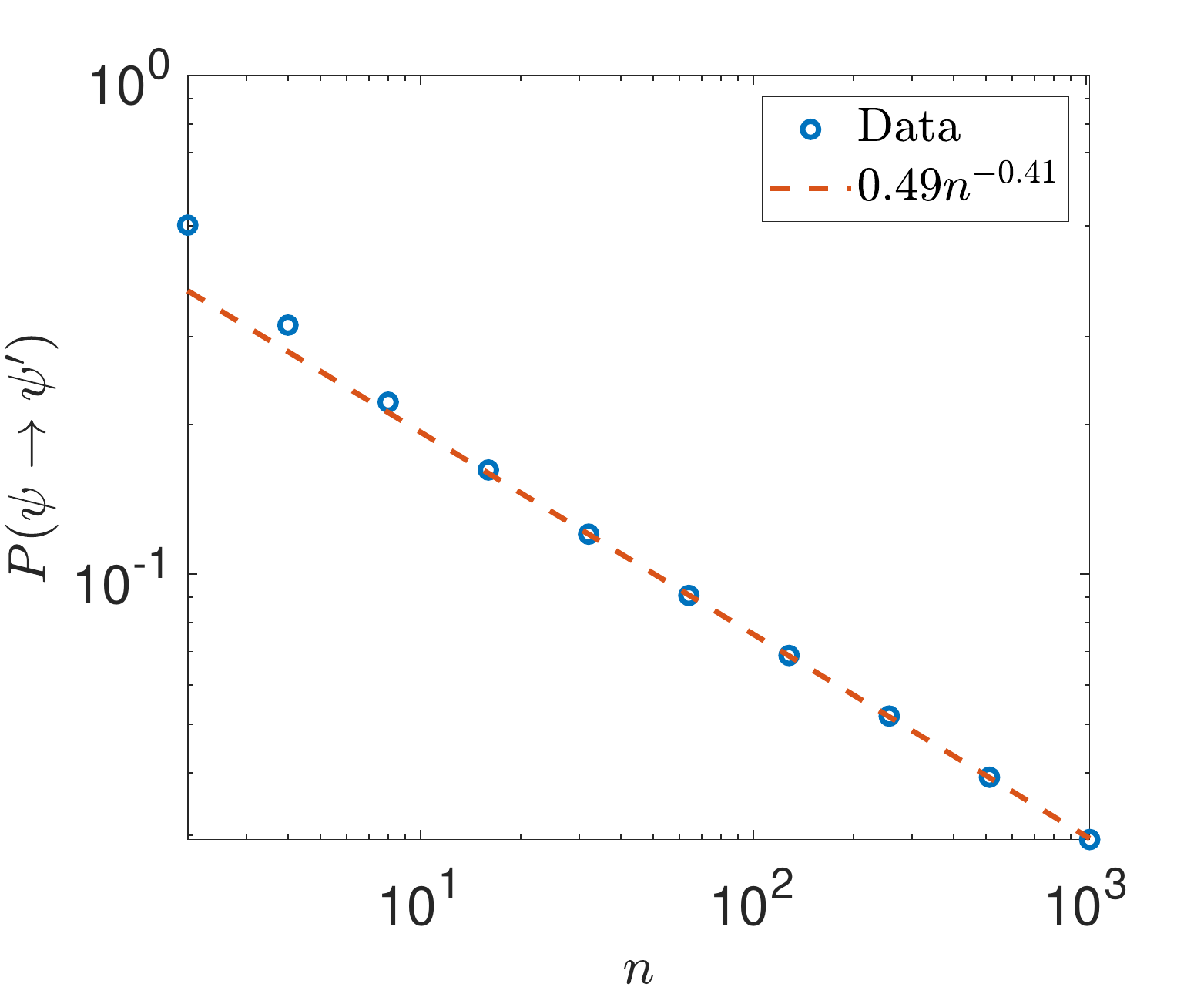}
	\caption{The probability of incoherent convertibility $P(\psi\prec \varphii)$ for $\psi$ and $\varphii$ independently chosen from the uniform distribution on pure states of $\St_n$.  The fit (dotted line) confirms an algebraic decay~\eqref{eq:thetadecay} with exponent $\theta = 0.4052$. Here $n=2,4,8,\ldots, 1024$. }
	\label{fig:Nielsen} 
\end{figure}

\subsection{Majorization, persistence probabilities and the arcsine law}

We next turn our attention to the convergence  rate   of $P(\psi\prec\psi')$ to $0$.
For two random pure states ${\psi},{\varphii}$ in $\St_n$, the vector $\tilde{\delta}=(\tilde{\delta}_k)_{1\leq k\leq n}$ with
\begin{equation}\label{eq:randomSteps}
	\tilde{\delta}_k=\delta(\psi')^\downarrow_k-\delta(\psi)^\downarrow_k,
\end{equation}
defines a continuous random walk $(S_k)_{0\leq k\leq n}$, started at $S_0:=0$, with steps $\tilde{\delta}_k$'s,
\begin{equation}\label{eq:randomWalk}
S_k:=\sum_{j=1}^k \tilde{\delta}_j \qquad 1\leq k\leq n.
\end{equation}
Note that $S_n=0$ (the process is a random bridge).
The majorization condition can be interpreted as the persistence (above the origin) of the random walk. Hence,
\begin{equation*}
	P(\psi\prec \psi') =P\left(\min_{1\leq k\leq n}S_k\geqslant 0\right).
\end{equation*}
Persistence probabilities for random processes have been widely studied in statistical physics and probability. Certain exactly solvable models (that include symmetric random walks, classical random bridge, integrated random walks, etc.), and numerical study of many other models, showed that in the general case the persistence probability above the origin decays \emph{algebraically} as $b n^{-\theta}$, for large $n$.  The so-called \emph{persistence exponents}  $\theta$ of a process, is typically very difficult to compute explicitly if the process is not Markov, although $\theta$ is believed to be distribution free within a universality class. 

It is therefore natural to expect that $P(\psi\prec \psi')$ also decays to zero as a power of $n$. In \figurename{~\ref{fig:Nielsen}} we show the results of numerical estimates of $P(\psi\prec \psi')$ (obtained from $10^6$ realisations of $\psi$, $\psi'$) for increasing values of $n$. The plot in logarithmic scale shows quite convincing evidence that 
\be
P(\psi\prec \psi') \sim b n^{-\theta}.
\label{eq:thetadecay}
\ee
 The value of the persistence exponent obtained from a numerical fit is $\theta=0.4052$, with $\sigma_\theta=0.0028$.

Note that the process $(S_k)_{0\leq k\leq n}$ is not Markov, and quite different from most familiar discrete-time random processes (the steps $\tilde{\delta}_k$'s in~\eqref{eq:randomSteps} are neither independent and identically distributed, as in a classical random walk,  nor distribution-invariant under permutations, as in a classical random bridge). One can nevertheless try to compute certain statistics related to the persistence of $S_k$, for instance the \emph{time spent above the origin}. Denote this time by $N_n:=\#\{k\leq n\colon S_k\geq0\}$. For classical symmetric random walks, the statistics of  $N_n$ is universal, and its limit is the well-known \emph{arcsine law}~\cite{Sparre-Andersen}. Surprisingly, numerical results (see \figurename{~\ref{fig:clauses}}) show that the fraction of time spent above $0$ by $S_k$ is also asymptotically described by the arcsine law,

		\be
		\lim_{n\to\infty}P\left(\frac{N_n}{n}\leq t\right)= \frac{2}{\pi}\arcsin\left(\sqrt{t}\right).
		\ee

\begin{figure}[t]
	\includegraphics[width=1\columnwidth]{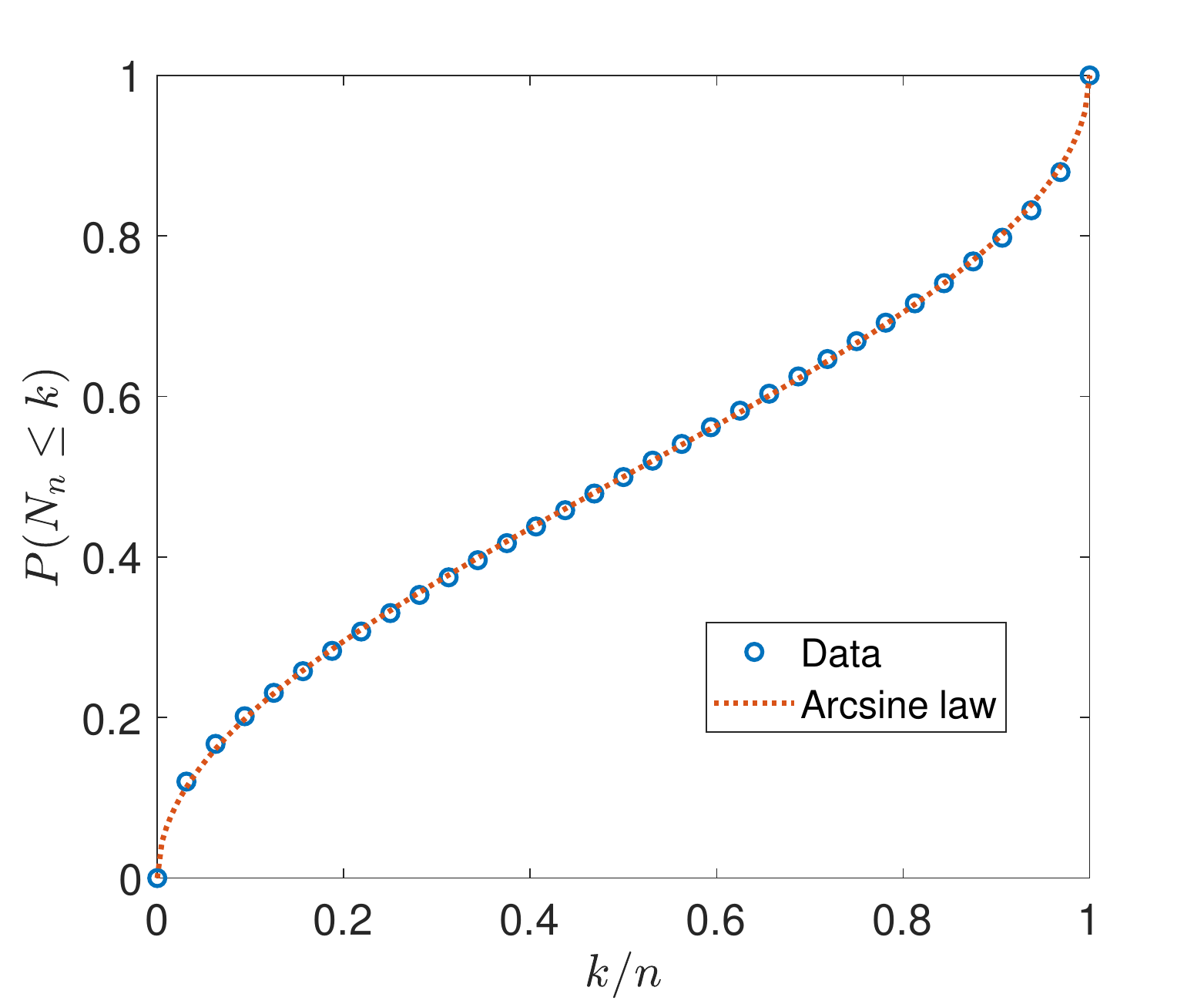}
	\caption{Probability distribution of $N_n$, the time spent by the
process $(S_k)$ above the origin, compared with the arcsine law. Here $n=32$.} 
	\label{fig:clauses}
\end{figure}

\subsection{Limit distribution of the maximal success probability of IO-conversion}
The maximal success probability of IO-conversion of state $\psi$ into $\varphii$ is $\Pi(\mu,\mu')$, where $\mu = \delta(\psi)$, $\mu'=\delta(\varphii)$ are the diagonal entries of $\psi$, $\varphii$ (Theorem~\ref{thm:probconv}). Theorem~\ref{thm:Pvanishing} can be rephrased as the statement that if $\mu,\mu'$ are independent uniform points in $\Delta_{n-1}$, then 
$$
P(\Pi(\mu,\mu')=1)\to0,\quad\text{ as $n\to\infty$.}
$$ 
In our previous work~\cite{majorization} on the LOCC-convertibility for random states, we conjectured a connection between the asymptotic fluctuations of the smallest component of random probability vectors $\lambda,\lambda'$ and the scaling limit of the variable $\Pi(\lambda,\lambda')$, when $\lambda$, $\lambda'$ are independent spectra of fixed-trace Wishart random matrices. Translated in this setting, the precise statement is that, if for some scaling constants $a_n$, $b_n$, the variable
\be
a_n\frac{\mu_{n}^{\downarrow}}{\mathbb{E}[\mu_{n}^{\downarrow}]}+b_n
\label{eq:scaling_min}
\ee
has a nontrivial limit in distribution, then, with the same constants,
\be
P(a_n\Pi(\mu,\mu')+b_n\leq p) 
\label{eq:scaling_Vid}
\ee
has a nontrivial limit as $n\to\infty$.

The smallest component $\mu_n^{\downarrow}=\delta(\psi)^{\downarrow}_n$ has probability density
$$
p_n(x)=n^2(1-nx)^{n-1}1_{0\leq x\leq 1/n}.
$$
The average and variance of $\mu_n^{\downarrow}$ are
$$
\mathbb{E}[\mu_n^{\downarrow}]=\frac{1}{n(n+1)},\quad
\operatorname{Var}[\mu_n^{\downarrow}]=\frac{1}{n(n+1)^2(n+2)}.
$$
Hence, the fluctuations of $\mu_n^{\downarrow}$ relative to the mean are asymptotically bounded,
$\operatorname{Var}[\mu_n^{\downarrow}]^{\frac{1}{2}}/\mathbb{E}[\mu_n^{\downarrow}]=O(1)$, and therefore we can take $a_n=1$ and $b_n=0$ in~\eqref{eq:scaling_min}. The conjectural statement~\eqref{eq:scaling_Vid} in this case says that the distribution function
\begin{equation}\label{eq:distF}
	F_n(p)=P(\Pi(\mu,\mu')\leqslant p)
\end{equation}
should have a scaling limit. Indeed, we found numerically that, for large $n$, the function $F_n(p)$ tends to a limit distribution, as shown in \figurename{~\ref{fig:distPiMarkovChain}}.

Here we push further our previous conjecture and we propose that as $n\to\infty$
the distribution of the random variable $\Pi(\mu,\mu')$ is determined by the
asymptotic behaviour of the smallest components of $\mu$, $\mu'$  only. Any
fixed block of the order statistics  $n^2\mu^{\downarrow}_{n-j+1}$,
$j=1,\dots,k$ is asymptotic to the first $k$ components of a Poisson process
$(V_{j})_{j\geq1}$; similarly, $n^2\mu^{\downarrow}_{n-j+1}$ is asymptotic to
its own independent copy $(V_{j}')_{j\geq1}$. Hence, we conjecture that
$\Pi(\mu,\mu')$ is asymptotically distributed as 
\be
\Pi^{\infty}(V,V'):=\inf_{k\geq1}\frac{\sum_{j=1}^kV_j}{\sum_{j=1}^kV_j'},
\label{eq:limit_Vid}
\ee
where  $(V_{j})_{j\geq1}$ and  $(V_{j}')_{j\geq1}$ are two independent copies of a Poisson process with rate $1$ (i.e. point processes with independent exponential spacings). In formulae, if we denote by $F_{\infty}(p):=P(\Pi^{\infty}(V,V')\leqslant p)$ the distribution function of $\Pi^{\infty}(V,V')$, we claim that
\be
\lim_{n\to\infty}F_n(p)=F_{\infty}(p).
\label{eq:conj_lim}
\ee
We report in \figurename{~\ref{fig:distPiMarkovChain}} the result of numerical simulations of $\Pi(\mu,\mu')$  and $\Pi^{\infty}(V,V')$ on samples of $5\cdot 10^5$ pairs of random probability vectors $\mu$, $\mu'$ in $\Delta_{n-1}$, and pairs of random processes $V$ and $V'$. The agreement between the corresponding distributions $F_n(p)$ for large $n$, and $F_{\infty}(p)$ is quite convincing of the correctness of~\eqref{eq:limit_Vid}-\eqref{eq:conj_lim}. 
\begin{figure}[t]
	\centering
	\includegraphics[width=1\columnwidth]{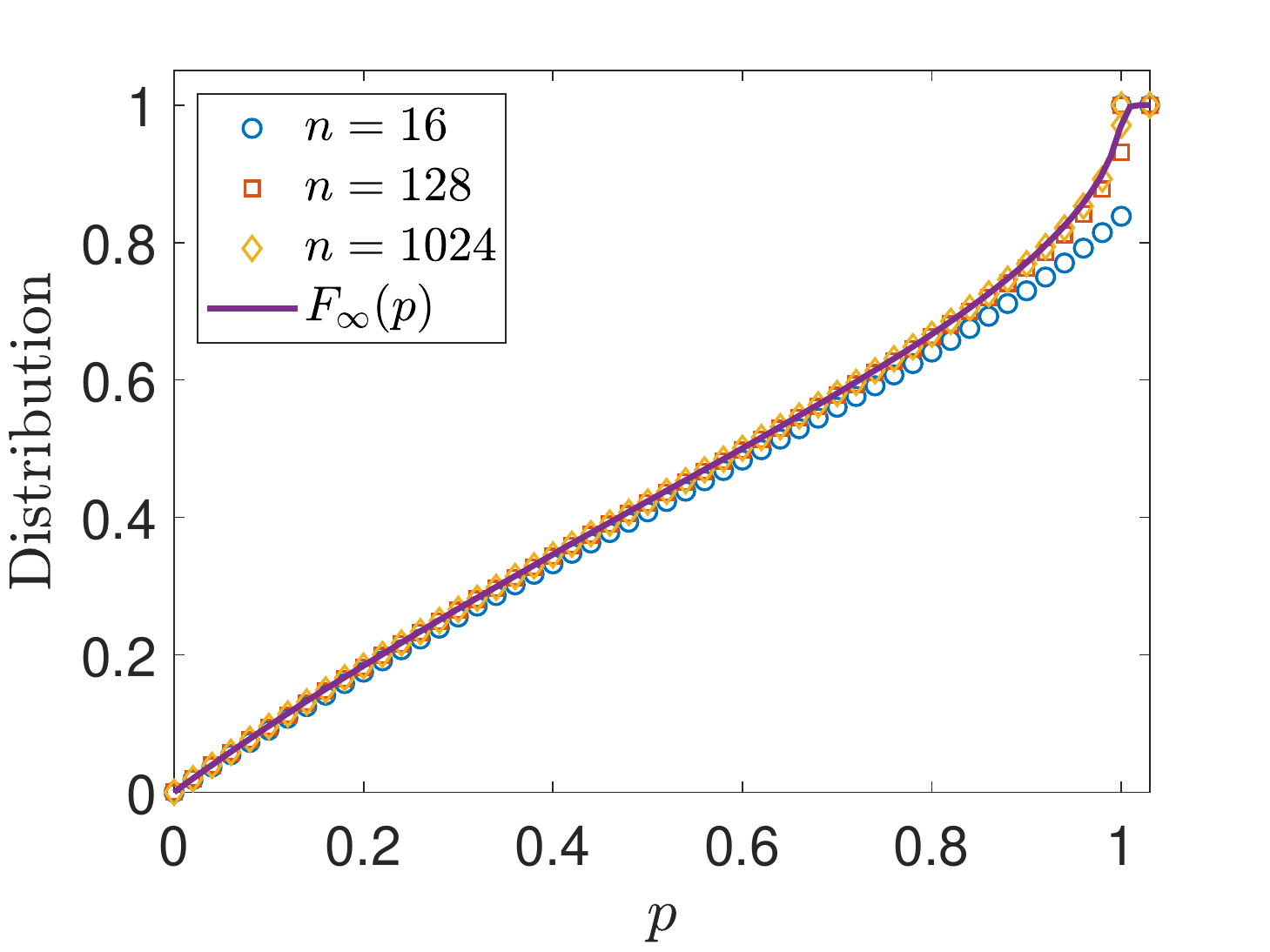}
	\caption{Distribution $F_n(p)$ for various values of $n$ vs the limit distribution $F_{\infty}(p)$ (distribution functions computed from numerical simulations).}
	\label{fig:distPiMarkovChain}
\end{figure}

\section{Concluding remarks}
\label{sec:conclusion}
\subsection{Likelihood of comparability in algebraic combinatorics}
In this paper we proved that the probability that two independent random points uniformly distributed in the unit simplex are in the majorization relation is asymptotically zero as $n\to\infty$. A similar question in the discrete setting was posed back in 1979 by MacDonald~\cite[Ch.1.1, Ex.18]{Macdonald}: for two integer partitions of $n$, chosen  uniformly at random, and independently, is it true that the probability that they are in majorization relation (a.k.a. \emph{dominance order}) is zero as $n\to\infty$? In 1999, Pittel~\cite{Pittel99}  proved the positive answer to Macdonald's question. In the proof of Theorem~\ref{thm:Pvanishing}  we emulated the main ideas exposed in~\cite{Pittel99}.  We mention however a simplification that occurs in the continuous setting. Pittel considered the first $k$ conditions for majorization involving the $k$ largest components of random integer partitions. They are asymptotic to Markov chains $(W_j)_{j\geq1}$ of Lemma~\ref{lem:Markov} with double-exponential spacings. In the continuous setting of random points in the simplex, it is fairly easy to obtain the asymptotics of the smallest components. Hence, here we considered instead the first $k$ conditions~\eqref{eq:maj} involving the smallest components of $\mu^{\downarrow}$, and this reduces the problem to the persistence probability of an integrated random walk $I_k=\sum_{j=1}^k\sum_{i=1}^j\widetilde{X}_i$, where the increments $\widetilde{X}_i$ have two-sided exponential distribution. This choice makes the analysis simpler compared to the discrete setting for integer partitions.
\par
\subsection{Entanglement theory and LOCC-convertibility}
\label{sub:Entanglement}
There is a difference between the statement proved in Theorem
\ref{thm:Pvanishing} for the theory of coherence and Nielsen's conjecture for
entanglement theory. The LOCC-convertibility criterion for bipartite systems is
the majorization relation for spectra of reduced density matrices. In the random
setting, those spectra are not uniformly distributed in the simplex; instead, they follow a random-matrix-type density in $\Delta_{n-1}$
$$
p_{\mathrm{RMT}}(x)=c_{n,m}\prod_{1\leq i<j\leq
n}(x_i-x_j)^2\prod_{k=1}^nx_k^{m-n}1_{x\in\Delta_{n-1}},
$$
where $n$ and $m$ are the dimensions of the subsystems (see~\cite[Eq.~(30)]{majorization}).
In particular, for a point distributed according to $p_{\mathrm{RMT}}$ it is no longer true that the largest/smallest components are 
asymptotically described by point processes with independent spacings (their 
statistics is given instead by scaling limits at the edges of random matrices,
 known as Airy and Bessel point processes), and this complicates considerably the
analysis.   
\par
We also note that the proof that $P(\mu\prec \mu')\to 0$ as $n\to\infty$
presented here can be adapted to the
case where $\mu,\mu'$ are independent copies of points in the unit simplex picked
according to a more general Dirichlet distribution
$$
p_{\mathrm{Dir}(\alpha)}(x)=\frac{\Gamma(n\alpha)}{\Gamma(\alpha)^n}
\prod_{k=1}^nx_k^{\alpha-1}1_{x\in\Delta_{n-1}}.
$$
It would be interesting to see if this
distribution appears naturally in the theory of random quantum states. 

\subsection{Other Resource Theories}
Majorization criteria play an important role also in other resource theories, such as the resource theory of purity, which is also closely connected to the resource theory of coherence (see e.g.~\cite{Streltsov18}). It would be interesting to investigate the applicability of our methods to these other scenarios, even beyond the pure state case.

\section{Acknowledgments}
\textit{Acknowledgments - } 
PF, GF and GG are partially supported by Istituto Nazionale di Fisica Nucleare (INFN) through the project ``QUANTUM".   
GF is supported by the Italian Ministry MIUR-PRIN project ``Mathematics of active materials: From mechanobiology to smart devices '' and by the FFABR research grant. FDC, PF, GF and GG are partially supported by the Italian National Group of Mathematical Physics (GNFM-INdAM).
FDC wishes to thank Ludovico Lami and Vlad Vysotskyi for valuable correspondence.

\appendix

\begin{widetext}

\section{Order statistics of i.i.d. random variables and Markov property}
We collect here a series of more or less known results about order statistics of independent random variables. 
In the following, $X_1,X_2,\dots$ are independent and identically distributed (i.i.d.) random variables with distribution function $F(x):=P(X_1\leq x)$. We always assume that they have a density $f(x)=F'(x)$. 
\par
For a finite family $X_1,X_2,\dots,X_n$, the \emph{order statistics} $X_k^{\downarrow}$, $k\leq n$, are the rearrangements of the variables in nonincreasing order, i.e. $X^{\downarrow}_1\geq X_2^{\downarrow}\geq\cdots\geq X_n^{\downarrow}$. Of course, the order statistics are \emph{not} i.i.d. random variables.

Under the previous assumptions on the distribution of the $X_i$'s, the order statistics have a density. The following lemma gives the explicit formulae that we need for our calculations.
\begin{lem} Let $X_1,X_2,\dots,X_n$ as above. Then,
\begin{itemize}
\item[i)] The density of $X^{\downarrow}_k$  is
\be
f_{X^{\downarrow}_k}(x_k)=\frac{n!}{(n-k)!(k-1)!}F(x_k)^{n-k}f(x_k)(1-F(x_k))^{k-1};
\label{eq:density_1}
\ee
\item[ii)] The joint density of $(X^{\downarrow}_k,X^{\downarrow}_l)$, for $k\leq l$ is
\be
f_{X^{\downarrow}_k,X^{\downarrow}_l}(x_k,x_l)=\frac{n!}{(n-l)!(l-k-1)!(k-1)!}F(x_l)^{n-l}f(x_l)(F(x_k)-F(x_l))^{l-k-1}f(x_k)(1-F(x_k))^{k-1}
\ee
for $x_k\geq x_l$, and zero otherwise;
\item[iii)] The joint density of the $k$ largest variables $(X^{\downarrow}_1,X^{\downarrow}_2,\ldots,X^{\downarrow}_{k})$ is
\be
f_{X^{\downarrow}_1,X^{\downarrow}_2,\dots X^{\downarrow}_{k}}(x_1,x_2,\dots,x_{k})=\frac{n!}{(n-k)!}f(x_1)f(x_2)\cdots f(x_{k})F(x_{k})^{n-k},
\label{eq:density_12k1}
\ee
for $x_1\geq x_2\geq\cdots\geq x_{k}\geq 0$, and zero otherwise;
\item[iv)] The joint density of the $k$ smallest variables $(X^{\downarrow}_{n-k+1},\ldots,X^{\downarrow}_{n-1}, X^{\downarrow}_n)$ is
\be
f_{X^{\downarrow}_{n-k+1}, \dots, X^{\downarrow}_{n-1}, X^{\downarrow}_n}(x_{n-k+1}, \dots, x_{n-1}, x_n)=\frac{n!}{(n-k)!} \left(1-F(x_{n-k+1})\right)^{n-k} f(x_{n-k+1})\cdots f(x_{n-1}) f(x_n) ,
\label{eq:density_12k2}
\ee
for 
$x_{n-k+1} \ge x_{n-k+2} \ge\cdots\ge  x_n \ge  0$, and zero otherwise.
\end{itemize}
\end{lem}
\begin{proof}
The proof is rather elementary  (see, e.g.~\cite{ABN08}). We sketch only the proof of Part i) to give a flavour of the type of arguments involved. The probability that $X^{\downarrow}_k$ is in $x_k$, is the probability that, among $X_1,\ldots, X_n$:  one is in $x_k$ (this gives a factor $f(x_k)$); exactly $(k-1)$ are larger than $x_k$ (this gives the factor $(1-F(x_k))^{k-1}$);  the remaining $(n-k)$ variables are smaller than $x_k$ (corresponding to the factor $F(x_k)^{n-k}$). There are $n\binom{n-1}{k-1}$ ways to partition the $n$ variables in that manner. 
\end{proof}
\begin{prop}
\label{prop:Mchain} Let $X_1,X_2,\ldots,X_n$ be as above. Then,
\begin{itemize}
\item[i)]The vector $(X^{\downarrow}_1,X^{\downarrow}_2,\ldots,X^{\downarrow}_n)$ forms an inhomogeneous (finite) Markov chain with initial density
\be
f_{X^{\downarrow}_1}(x)=nF(x)^{n-1}f(x),\label{eq:initial_dens}
\ee
and transition densities given by
\be
f_{X^{\downarrow}_{k+1}|X^{\downarrow}_{k}}(y|x)=
\begin{cases}
(n-k)\dfrac{F(y)^{n-k-1}}{F(x)^{n-k}}f(y) &\text{if $y\leq x$}\\
0&\text{otherwise}.
\end{cases}\label{eq:trans_densS}
\ee
\item[ii)] The vector $(X^{\downarrow}_n,X^{\downarrow}_{n-1},\ldots,X^{\downarrow}_1)$ forms an inhomogeneous (finite) Markov chain with initial density
\be
f_{X^{\downarrow}_n}(x)=n(1-F(x))^{n-1}f(x),\label{eq:initial_dens2}
\ee
and transition densities given by
\be
f_{X^{\downarrow}_{n-k}|X^{\downarrow}_{n-k+1}}(y|x)=
\begin{cases}
(n-k)\dfrac{(1-F(y))^{n-k-1}}{(1-F(x))^{n-k}}f(y) &\text{if $y\geq x$}\\
0&\text{otherwise}.
\end{cases}\label{eq:trans_densS2}
\ee
\end{itemize}
\end{prop}
\begin{proof} We  prove Part i). The density~\eqref{eq:initial_dens} is a specialisation of~\eqref{eq:density_1}. 
For $k\leq l$, the conditional density of $X^{\downarrow}_l$ given $X^{\downarrow}_k$ is
\begin{align}
f_{X^{\downarrow}_l|X^{\downarrow}_k}(x_l|x_k)&=\frac{f_{X^{\downarrow}_k,X^{\downarrow}_l}(x_k,x_l)}{f_{X^{\downarrow}_k}(x_k)}=\frac{(n-k)!}{(n-l)!(l-k-1)!}\frac{F(x_l)^{n-l}}{F(x_k)^{n-k}}(F(x_k)-F(x_l))^{l-k-1}f(x_l),
\end{align}
for $x_k\geq x_l$, and zero otherwise. In particular, for $l=k+1$, we get~\eqref{eq:trans_densS}.
Similarly, from~\eqref{eq:density_12k1}, we have, for all $x_1\geq x_2\geq\cdots\geq x_k$,
 \begin{align}
f_{X^{\downarrow}_{k+1}|X^{\downarrow}_1,X^{\downarrow}_2\ldots,X^{\downarrow}_k}(x_{k+1}|x_1,x_2,\ldots,x_k)&=
\frac{f_{X^{\downarrow}_1,X^{\downarrow}_2,\dots X^{\downarrow}_{k+1}}(x_1,x_2,\dots,x_{k+1})}{f_{X^{\downarrow}_1,X^{\downarrow}_2,\dots X^{\downarrow}_{k}}(x_1,x_2,\dots,x_{k})}=(n-k)\dfrac{F(x_{k+1})^{n-k-1}}{F(x_k)^{n-k}}f(x_{k+1}),
\end{align}
for $x_{k+1}\leq x_k$, and zero otherwise. Hence, we have  proved that
\be
f_{X^{\downarrow}_{k+1}|X^{\downarrow}_1,X^{\downarrow}_2\ldots,X^{\downarrow}_k}(x_{k+1}|x_1,x_2,\ldots,x_k)=f_{X^{\downarrow}_{k+1}|X^{\downarrow}_k}(x_{k+1}|x_k).
\ee
\emph{Mutatis mutandi} we can prove Part ii).
\end{proof}
\begin{rmk}
\label{rmk:distr}
The previous formulae for the densities can be rephrased in terms of the distribution functions:
\begin{align}
F_{X^{\downarrow}_1}(x)&=\int_{-\infty}^xf_{X^{\downarrow}_1}(z)dz=F(x)^n,\\
F_{X^{\downarrow}_{k+1}|X^{\downarrow}_{k}}(y|x)&=\int_{-\infty}^yf_{X^{\downarrow}_{k+1}|X^{\downarrow}_{k}}(z|x)dz=\left(\frac{F(\min(y,x))}{F(x)}\right)^{n-k}\\
F_{X^{\downarrow}_n}(x)&=1-\int_{x}^{+\infty}f_{X^{\downarrow}_n}(z)dz=1-(1-F(x))^n\\
F_{X^{\downarrow}_{n-k}|X^{\downarrow}_{n-k+1}}(y|x)&=1-\int_{y}^{\infty}f_{X^{\downarrow}_{n-k}|X^{\downarrow}_{n-k+1}}(z|x)dz=1-\left(\frac{1-F(\max(y,x))}{1-F(x)}\right)^{n-k}.
\end{align}
\end{rmk}
\section{Proof of Proposition~\ref{prop:Markov}}
\label{app:proof_propMarkov}
  \begin{proof}
We first recall a standard representation for the uniform distribution in $\Delta_{n-1}$ in terms of i.i.d. exponential random variables, and the classical asymptotic distributions of the extreme values for exponential random variables.
 \begin{lem} Let $X_1,X_2,\ldots$ be independent exponential random variables with rate $1$, i.e. $P(X\leq x)=1-\e^{-x}$. Then, the vector
 \be
  \left(\mu_1,\mu_2,\ldots,\mu_n\right):=\left(\frac{X_1}{\sum_{i=1}^nX_i},\frac{X_2}{\sum_{i=1}^nX_i},\ldots,\frac{X_n}{\sum_{i=1}^nX_i}\right)
  \label{eq:YZ}
\ee
 is uniformly distributed in $\Delta_{n-1}$.
 \end{lem}
 \begin{lem}
 \label{lem:asymptotics} If $F(x)=1-\e^{-x}$, then
\begin{align}
  \lim_{n\to\infty} 1-(1- F(u/n))^n&=1-\exp(-{u})\quad \text{(exponential distribution)}.\label{eq:asymp_exp}\\
  \lim_{n\to\infty}F(\log n+u)^n&=\exp(-\e^{-u})\quad \text{(Gumbel distribution)\label{eq:asymp_Gumbel}}
\end{align}
\end{lem}
Let $\mu=(\mu_1,\mu_2,\ldots,\mu_n)$ defined as in~\eqref{eq:YZ} be a uniform point on $\Delta_{n-1}$. 
Combining Proposition~\ref{prop:Mchain}, Remark~\ref{rmk:distr}, and the formula~\eqref{eq:asymp_exp}, we see that  for any fixed $k\geq1$, 
 $$
  \left(nX^{\downarrow}_{n},nX^{\downarrow}_{n-1},\ldots,nX^{\downarrow}_{n-k+1}\right)
 $$
  converges in distribution to the first $k$ components $(V_1,V_2,\ldots,V_k)$ of the  time-homogeneous Markov chain $(V_j)_{j\geq1}$ with density of $V_1$ and transition density 
$$
f_{V_1}(v)=\exp(-v)1_{v\geq0},\qquad
f_{V_{j+1}|V_j}(u|v)=\exp(v-u)1_{u\geq v},
$$
respectively.  To show the convergence for the order statistics of $\mu$, we simply observe that the vector
   $$
    \left(n^2\mu^{\downarrow}_{n-j+1}\right)_{1\leq j\leq k}=    \left(\frac{n}{\sum_{i=1}^nX_i}nX^{\downarrow}_{n-j+1}\right)_{1\leq j\leq k}
    $$
    has the same limit distribution of $\left(nX^{\downarrow}_{n-j+1}\right)_{1\leq j\leq k}
$.  (Recall that $\mathbb{E} [\sum_{i=1}^nX_i ]=n$; hence the factor  
 $n^{-1}\sum_{i=1}^nX_i$ converges to $1$ by the  law of large numbers.) 

\par
Similarly, from Proposition~\ref{prop:Mchain}, Remark~\ref{rmk:distr}, and the asymptotic formula~\eqref{eq:asymp_Gumbel}, we deduce that
 $$
  \left(X^{\downarrow}_1-\log n,X^{\downarrow}_2-\log n,\ldots,X^{\downarrow}_k-\log n\right)
  $$
  converges in distribution to the first $k$ components $(W_1,W_2,\ldots,W_k)$ of the  time-homogeneous Markov Chain $(W_j)_{j\geq1}$ with density of $W_1$ and transition density 
$$
f_{W_1}(w)=\exp(-\e^{-w}-w),\qquad
f_{W_{j+1}|W_j}(u|w)=\exp(\e^{-w}-\e^{-u}-u)1_{u\leq w},
$$
respectively. 
Denote by $\mu^{\downarrow}$ the decreasing rearrangement of $\mu$. For any $k$, we want to show that
  $$
  \left(n\mu^{\downarrow}_1-\log n,n\mu^{\downarrow}_2-\log n,\ldots,n\mu^{\downarrow}_k-\log n\right)
  $$
  converges in distribution to
  $ (W_1,W_2,\dots,W_k)$. We write
   $$
    \left(n\mu^{\downarrow}_j-\log n\right)_{1\leq j\leq k}=    \left(\frac{n}{\sum_{i=1}^nX_i}\left(X^{\downarrow}_j-\log n\right)+\left(\frac{n}{\sum_{i=1}^nX_i}-1\right)\log n\right)_{1\leq j\leq k},
    $$
and we want to show that this vector has  the same limit distribution of  $\left(X^{\downarrow}_j-\log n\right)_{1\leq j\leq k}$,  as $n\to\infty$. The factor  
 $n^{-1}\sum_{i=1}^nX_i$ converges to $1$ by the  law of large numbers.
 For all $\epsilon>0$,
\begin{align*}
 P\left(\left|\frac{n}{\sum_{i=1}^nX_i}-1\right|>\frac{\epsilon}{\log n}\right)
 &= 
 P\left(\sum_{i=1}^nX_i<\frac{n\log n}{\log n+\epsilon}\quad\text{or}\quad \sum_{i=1}^nX_i>\frac{n\log n}{\log n-\epsilon} \right)\\
& =
 P\left(\sum_{i=1}^nX_i<n-\frac{n\epsilon}{\log n+\epsilon}\quad\text{or}\quad \sum_{i=1}^nX_i>n+\frac{n\epsilon}{\log n-\epsilon} \right)\\
&\leq
 P\left(\left|\sum_{i=1}^nX_i-n\right|>\min\left\{\frac{n\epsilon}{\log n+\epsilon},\frac{n\epsilon}{\log n-\epsilon}\right\} \right).
 \end{align*}
Recall that $\operatorname{Var} [\sum_{i=1}^nX_i] =n$. Assuming $n>\exp(\epsilon)$, and using Chebyshev's inequality we can estimate
\begin{align*}
 P\left(\left|\frac{n}{\sum_{i=1}^nX_i}-1\right|>\frac{\epsilon}{\log n}\right)
  &\leq \frac{\operatorname{Var}\left[\sum_{i=1}^nX_i\right]}{\left(\frac{n\epsilon}{\log n+\epsilon}\right)^2}
= \frac{1}{n\epsilon^2}\left(\log n+\epsilon\right)^2.
\end{align*}
Hence,  $\left(\frac{n}{\sum_{i=1}^nX_i}-1\right)\log n$ converges to $0$ in probability as $n\to\infty$. 
 \end{proof}

\section{Vanishing of the persistence probability above the origin of the IRW. Proof of  Claim~\eqref{eq:ppvanish} in Theorem~\ref{thm:Pvanishing}}
\label{app:proof_claim}

We want to prove that the persistence probability asymptotically vanishes,
$$ 
\lim_{k\to\infty} P\left(\min_{1\leq j\leq k}I_j\geq0\right) = 0.
$$
Notice that 
$$
\lim_{k\to\infty} P\left(\min_{1\leq j\leq k}I_j\geq0\right)  = P\left(\inf_{k \geq1} I_k \geq 0 \right) \leq 
P\left(\liminf_{k \to\infty} I_k \geq 0 \right) = P\left(\liminf_{k\to\infty}\frac{{I}_k}{k\log k}\geq0\right).
$$
Therefore, it is sufficient to show that 
\be
P\left(\liminf_{k\to\infty}\frac{{I}_k}{k\log k}\geq0\right)=0,
\label{eq:prop_thesis_re}
\ee
and this follows from the Lindeberg-Feller central limit theorem as we outline now.

Denote by $A$ the event in~\eqref{eq:prop_thesis_re}. 
\begin{claim}
\label{claim:0-1}
$P(A)\in\{0,1\}$.
\end{claim}

The proof of the Claim is almost verbatim the proof given by Pittel~\cite{Pittel99}.
For the event
$$
A=\left\{\liminf_{k\to\infty}\frac{{I}_k}{k\log k}\geq0\right\},
$$
we want to show that $P(A)\in\{0,1\}$. The key observation here is that the probability of the event $A$ does not depend on the variables of $\widetilde{X}_1,\widetilde{X}_2,\ldots,\widetilde{X}_J$, no matter how large, albeit finite, $J$ is. 
Indeed, let $\widetilde{V}_k(J)=\sum_{j={J+1}}^k\widetilde{X}_j$, for $k>J$ and $I_k(J)=\sum_{j=J+1}^k \widetilde{V}_j(J)$ for $k>J$ as well. Then, 
$$
I_k-I_k(J)=\sum_{j=1}^k\widetilde{V}_j-\sum_{j={J+1}}^k \widetilde{V}_j(J)=\sum_{j=1}^k j \widetilde{X}_{k-j+1} - \sum_{j=1}^{k-J}j \widetilde{X}_{k-j+1} =\sum_{j=k-J+1}^kj\widetilde{X}_{k-j+1}.
$$

Therefore, almost surely 
$$
\lim_{k\to\infty}\frac{1}{k\log k}\left|I_k-I_k(J)\right|=0,\quad\text{for all $J$}.
$$

So, denoting
$$
A_J=\left\{\liminf_{k\to\infty}\frac{{I}_k(J)}{k\log k}\geq0\right\},
$$
we can write for the symmetric difference $A\triangle A_J$ of the events $A$ and $A_J$,
$$
P\left(A\triangle A_J\right)=0,\quad\text{for all $J$}.
$$
Now observe that $A_J$ is measurable with respect to $(\widetilde{X}_{j})_{j>J}$. (Informally, the event $A_J$ does not involve the first $J$ variables $\widetilde{X}_{1},\ldots,\widetilde{X}_{J}$). Then, writing $\mathrm{a.a.}$ for ``almost always'' and  $\mathrm{i.o.}$ for ``infinitely often'',
$$
A_{\infty}=\liminf_{J}A_J=\bigcup_{J\geq1}\bigcap_{m\geq J} A_m=\left\{A_J\;  \mathrm{a.a.}\right\}
$$
is a tail-event, and 
\begin{align*}
P\left(A\triangle A_{\infty}\right)&=P\left(A\cap A_{\infty}^c\right)+P\left(A^c\cap A_{\infty}\right) =P\left(A\cap A_{J}^c\;\mathrm{i.o.}\right)+P\left(A^c\cap A_{J}\; \mathrm{a.a.}\right)\\
&\leq\sum_{J\geq1}\left[P\left(A\cap A_{J}^c\right)+P\left(A^c\cap A_{J}\right)\right]=\sum_{J\geq1}P\left(A\triangle A_J\right)
=0.
\end{align*}
By the Kolmogorov 0-1 law, $P(A_{\infty})\in\{0,1\}$, so from the previous calculation we obtain $P(A)\in\{0,1\}$, as well.

Given Claim~\ref{claim:0-1}, we can now complete the proof if we show that $P(A)<1$. By the definition of $A$, to do so it suffices to show that
\be
\lim_{k\to\infty}P\left(\frac{{I}_k}{k\log k}\geq-b\right)<1,
\ee
for a constant $b>0$.
 Writing  ${I}_k=\sum_{j=1}^kj\widetilde{X}_{k-j+1}$ it is a routine matter to compute
$$
\mathbb{E}[{I}_k]=0,\,\, \operatorname{Var}[{I}_k]=2\sum_{j=1}^kj^2=\frac{k(k+1)(2k+1)}{3}=O(k^3).
$$
From this, one can check that the sequence
${I}_k$ satisfies the Lindeberg-Feller conditions, and thus ${I}_k/\sqrt{\operatorname{Var}[{I}_k]}$ converges in distribution to the standard Gaussian variable as $k\to\infty$. Hence,
\begin{align*}
P\left(\frac{{I}_k}{k\log k}\geq-b\right)=P\left(\frac{\sqrt{\operatorname{Var}[{I}_k}]}{k\log k}\frac{{I}_k}{\sqrt{\operatorname{Var}[{I}_k}]}\geq-b\right) \stackrel{k\to\infty}{\longrightarrow}\frac{1}{\sqrt{2\pi}}\int_0^{\infty}\e^{-x^2/2}dx=\frac{1}{2}<1.
\hfill\qed
\end{align*}

\end{widetext}

\end{document}